\documentclass{article}

\usepackage{microtype}
\usepackage{graphicx}
\usepackage{dirtytalk}
\usepackage{subfigure}
\usepackage{booktabs} % for professional tables

\usepackage{hyperref}

\usepackage[accepted]{icml2025}

\usepackage{amsmath}
\DeclareMathOperator{\E}{\mathbb{E}}
\usepackage{amssymb}
\usepackage{mathtools}
\usepackage{amsthm}
\usepackage{dirtytalk}

\usepackage{algorithm}

\usepackage[capitalize,noabbrev]{cleveref}

\theoremstyle{plain}
\newtheorem{theorem}{Theorem}[section]

\theoremstyle{definition}

\theoremstyle{remark}

\usepackage[textsize=tiny]{todonotes}

\icmltitlerunning{Submission and Formatting Instructions for ICML 2025}
\icmltitlerunning{Wasserstein-Barycenter Consensus for Cooperative MARL}

\begin{document}

\twocolumn[
\icmltitle{Wasserstein-Barycenter Consensus for \\ Cooperative Multi-Agent Reinforcement Learning}



\icmlsetsymbol{equal}{*}

\begin{icmlauthorlist}
\icmlauthor{Ali Baheri}{rit}
%\icmlauthor{Glen Berseth}{mila}
\end{icmlauthorlist}

\icmlaffiliation{rit}{Kate Gleason College of Engineering, Rochester Institute of Technology, Rochester, USA}
%\icmlaffiliation{mila}{Mila, Universit\'e de Montr\'eal, CIFAR Canada AI Chair, Montreal, Canada}

\icmlcorrespondingauthor{Ali Baheri}{akbeme@rit.edu}
%\icmlcorrespondingauthor{Glen Berseth}{glen.berseth@mila.quebec}

\icmlkeywords{Machine Learning, ICML, Multi-Agent Reinforcement Learning, Optimal Transport, Wasserstein Barycenter}

\vskip 0.3in
]



\printAffiliationsAndNotice{}  % leave blank if no need to mention equal contribution

\begin{abstract}
Cooperative multi-agent reinforcement learning (MARL) demands principled mechanisms to align heterogeneous policies while preserving the capacity for specialized behavior. We introduce a novel consensus framework that defines the team strategy as the entropic-regularized p-Wasserstein barycenter of agents’ joint state–action visitation measures. By augmenting each agent’s policy objective with a soft penalty proportional to its Sinkhorn divergence from this barycenter, the proposed approach encourages coherent group behavior without enforcing rigid parameter sharing. We derive an algorithm that alternates between Sinkhorn-barycenter computation and policy-gradient updates, and we prove that, under standard Lipschitz and compactness assumptions, the maximal pairwise policy discrepancy contracts at a geometric rate. Empirical evaluation for a cooperative navigation case study demonstrates that our OT-barycenter consensus outperforms an independent learners baseline in convergence speed and final coordination success.
\end{abstract}

\section{Introduction}

The success of cooperative multi–agent reinforcement learning (MARL) is based on effective coordination among agents whose individual decision making processes induce distinct probability distributions in states and actions \cite{yuan2023survey}. In domains as varied as autonomous vehicle platooning, distributed sensor networks, and multi‐robot manipulation, misalignment of agent policies can lead to suboptimal team performance, brittle coordination, and slow convergence \cite{busoniu2008comprehensive}. Classical approaches to MARL coordination, such as shared reward shaping, centralized critics, or hard parameter sharing, either incur prohibitive communication and computation costs or overly constrain agent heterogeneity, preventing the emergence of specialized roles \cite{yi2022learning,lyu2023centralized,zhao2023semi}.

Optimal transport (OT) is a powerful mathematical framework for comparing and interpolating probability distributions, offering a geometrically intuitive way to measure discrepancies and fuse information \cite{villani2009optimal,peyre2019computational}. Its application in single agent reinforcement learning has yielded significant advances, particularly in areas such as distributional RL, offline RL \cite{omuraoffline}, safe RL \cite{baheri2025wave,shahrooei2025risk}, and inverse RL \cite{baheri2023understanding}. The Wasserstein distance, a key OT metric, is particularly appealing as it provides a smooth and meaningful measure of divergence even between distributions with non-overlapping supports, a common scenario in policy space. Furthermore, the development of entropic regularization and the associated Sinkhorn divergence has led to computationally efficient, differentiable algorithms for approximating Wasserstein distances, making them amenable to gradient-based optimization in deep learning settings \cite{cuturi2013sinkhorn}. Despite these advances, the potential of OT to orchestrate cooperation in MARL has remained largely unexplored: existing OT‑based methods focus primarily on aligning belief or value distributions in single‑agent or competitive settings, but do not exploit the notion of a team “consensus” policy as an explicit barycenter of individual behaviors.

To address this gap, we propose a novel consensus mechanism in which the team’s collective strategy is formalized as the entropic‑regularized Wasserstein barycenter of all agents’ visitation measures over the joint state–action space. Rather than enforcing identical weights or centralized value estimation, each agent incurs a soft penalty proportional to its Wasserstein distance from this barycenter, thereby allowing agents to retain beneficial specialization while still gravitating toward coherent group behavior. We develop a fully differentiable algorithm that alternates between (i) computing the Sinkhorn barycenter via iterative Bregman projections and (ii) performing policy gradient updates augmented with a consensus regularizer. We analyze the resulting dynamics and prove that, under standard Lipschitz and compactness assumptions, the maximal pairwise Sinkhorn divergence contracts at a geometric rate, guaranteeing convergence to a common consensus distribution.

The remainder of this paper is organized as follows. Section 2 provides essential background on optimal transport. In Section 3, we introduce our Wasserstein-Barycenter Consensus framework and present theoretical guarantees. Section 4 evaluates our method on a cooperative multi-agent navigation task, comparing it against an independent learning baseline. Finally, Section 5 concludes the paper with a summary of our contributions and outlines potential directions for future work.

\noindent{\textbf{Related Work.}} Early efforts to induce coordination in cooperative MARL rely on parametric or value-factorization surrogates of consensus. Parameter–sharing schemes update a single network with trajectories from all agents, optionally injecting agent-identifiers to recover role specialization \cite{terry2020revisiting}. Centralized-training-with-decentralized-execution (CTDE) algorithms pursue a similar goal through a central critic (e.g., COMA’s counterfactual baseline) or a monotonic joint-value factorization (e.g., QMIX), thereby coupling agents only via gradients on a global objective \cite{foerster2018counterfactual}. Divergence-based regularizes extend this idea: DMAC penalizes the KL gap between local actors and a target mixture \cite{su2022divergence}; trust-weighted Q-learning adaptively scales updates by inter-agent credibility; and the recently-proposed MAGI framework coordinates policies around an imagined high-value goal state to mitigate mis-coordination \cite{wang2024reaching}. While empirically effective, these approaches align parameters or returns rather than the full visitation measures, and they provide at best heuristic or asymptotic coordination guarantees—limitations our proposed algorithm addresses by contracting the exact Sinkhorn divergence between agents.

\section{Preliminaries: Optimal Transport}

OT provides tools to compare probability distributions. The $p$-Wasserstein distance measures the cost of transporting mass between two distributions. We focus on its entropic-regularized version, the Sinkhorn divergence $W_{p, \epsilon}$ (where $\epsilon>0$ is the regularization parameter), which offers computational benefits. A key concept is the Wasserstein barycenter, $\mu^*$ , which is a measure that minimizes the average $p$-Wasserstein distance to a set of given measures $\left\{\mu_i\right\}_{i=1}^N$. Formally, the entropic-regularized $p$-Wasserstein barycenter is the solution to:
$$
\mu^*=\arg \min _{\mu \in \mathcal{P}(\mathcal{X})} \frac{1}{N} \sum_{i=1}^N W_{p, \epsilon}^p\left(\mu, \mu_i\right)+\epsilon K L(\mu \| \eta)
$$
where $K L(\cdot \| \eta)$ is a Kullback-Leibler divergence term with respect to a reference measure $\eta$. This paper leverages the Sinkhorn barycenter as a consensus point for agent policies in MARL.

\section{Methodological Approach}

To induce consensus via the Wasserstein barycenter, we begin by modeling each agent's stochastic policy as a probability measure in the joint state-action space $\mathcal{X}=$ $\mathcal{S} \times \mathcal{A}$. Concretely, let $\pi_i\left(a \mid s ; \theta_i\right)$ denote agent $i$'s policy parameterized by $\theta_i$, and let $\mu_i^{(t)}$ be its empirical visitation distribution over $\mathcal{X}$ at iteration $t$. We assume $\mathcal{X}$ is endowed with a ground metric $d\left((s, a),\left(s^{\prime}, a^{\prime}\right)\right)=\left\|s-s^{\prime}\right\|_2+\beta\left\|a-a^{\prime}\right\|_2$, where $\beta>0$ balances state and action discrepancies.
At each iteration, we seek a consensus measure $\mu^{*(t)}$ that minimizes the average entropic-regularized $p$-Wasserstein cost to the individual $\left\{\mu_i^{(t)}\right\}_{i=1}^N$. Formally, the barycenter is obtained by solving
$$
\mu^{*(t)}=\arg \min _{\mu \in \mathcal{P}(\mathcal{X})} \frac{1}{N} \sum_{i=1}^N W_{p, \varepsilon}^p\left(\mu, \mu_i^{(t)}\right)+\varepsilon \operatorname{KL}(\mu \| \eta),
$$
where $W_{p, \varepsilon}$ is the Sinkhorn divergence with entropic regularization parameter $\varepsilon>0, \mathrm{KL}(\cdot \| \eta)$ is a reference-measure penalty ensuring absolute continuity with respect to a base measure $\eta$, and $\mathcal{P}(\mathcal{X})$ denotes the space of probability measures on $\mathcal{X}$. We compute $\mu^{*(t)}$ via iterative Sinkhorn-barycenter updates: initializing a dual potential vector $u^{(0)}, v^{(0)}$ on a discretization of $\mathcal{X}$, we alternate
$$
u^{(\ell+1)}=\frac{\eta}{\sum_j K_{\varepsilon} u_j^{(\ell)} v_j^{(\ell)}}, \quad v^{(\ell+1)}=\left(\prod_{i=1}^N K_{\varepsilon}^{\top} u^{(\ell+1)}\right)^{-\frac{1}{N}}
$$
where $K_{\varepsilon}=e^{-D / \varepsilon}$ and $D_{i j}=d\left(x_i, x_j\right)^p$. Convergence of these updates to the unique Sinkhorn barycenter is guaranteed under compactness of $\mathcal{X}$ and strictly positive $\varepsilon$.

Once $\mu^{*(t)}$ is obtained, each agent's parameters $\theta_i$ are updated by performing a gradient ascent step on a regularized objective
$$
J_i\left(\theta_i\right)=\mathbb{E}_{(s, a) \sim \mu_i^{(t)}}[R(s, a)]-\lambda W_{p, \varepsilon}^p\left(\mu_i^{(t)}, \mu^{*(t)}\right)
$$
where $R(s, a)$ is the common team reward and $\lambda>0$ controls the strength of consensus enforcement. The gradient of the Wasserstein term with respect to $\theta_i$ is calculated by differentiating through the transport plan: if $\gamma_i^{*(t)}$ is the optimal coupling between $\mu_i^{(t)}$ and $\mu^{*(t)}$, then
$$
\nabla_{\theta_i} W_{p, \varepsilon}^p\left(\mu_i^{(t)}, \mu^{*(t)}\right)=\int_{\mathcal{X} \times \mathcal{X}} d(x, y)^p \nabla_{\theta_i} \log \pi_i\left(a \mid s ; \theta_i\right) \mathrm{d} \gamma_i^{*(t)}(x, y),
$$
which is estimated by sampling from $\mu_i^{(t)}$. The complete update rule thus reads
$$
\theta_i^{(t+1)}=\theta_i^{(t)}+\alpha\left(\hat{\nabla}_{\theta_i} \mathbb{E}[R]-\lambda \hat{\nabla}_{\theta_i} W_{p, \varepsilon}^p\right)
$$
with $\alpha>0$ the learning rate. To ensure computational tractability, support points for $\mu_i^{(t)}$ and $\mu^{*(t)}$ are drawn by mini-batch sampling, and sliced-Wasserstein approximations are used to reduce complexity. In addition, adaptive scheduling of $\lambda$ and $\varepsilon$ is used to warm start consensus (large $\lambda$) before relaxing. Specifically, the consensus weight $\lambda$ is initialized at a high value to encourage swift policy contraction and is gradually annealed to permit fine-grained individual behaviors near convergence. Concurrently, the entropic regularizer $\varepsilon$ is decreased over time, ensuring that early Sinkhorn iterations remain smooth and stable, while later iterations sharpen the barycenter estimate.

\begin{algorithm}[tb]
\caption{Wasserstein-Barycenter Consensus for Cooperative MARL}
\label{alg:wb-consensus}
\begin{algorithmic}[1]
\REQUIRE Number of agents $N$, entropic regularization $\varepsilon$, consensus weight $\lambda$, learning rate $\alpha$, metric $d$ on $\mathcal{S}\times\mathcal{A}$
\STATE Initialize policy parameters $\{\theta_i^{(0)}\}_{i=1}^N$
\FOR{$t = 0,1,2,\dots$}
    \STATE \textbf{Collect Trajectories:}
    \FOR{$i=1$ to $N$}
        \STATE Execute $\pi_i(\cdot\mid\cdot;\theta_i^{(t)})$ to sample trajectories
        \STATE Estimate empirical visitation measure $\mu_i^{(t)}$ over $\mathcal{S}\times\mathcal{A}$
    \ENDFOR

    \STATE \textbf{Compute Sinkhorn Barycenter:}
    \STATE Initialize dual potentials $u\leftarrow\mathbf{1},\,v\leftarrow\mathbf{1}$
    \REPEAT
        \STATE $u \leftarrow \displaystyle \frac{\eta}{\sum_j \exp\bigl(-D/\varepsilon\bigr)_{\,\cdot j}\,v_j}$
        \STATE $v \leftarrow \Bigl(\prod_{i=1}^N \exp\bigl(-D/\varepsilon\bigr)^{\!\top} u\Bigr)^{-1/N}$
    \UNTIL{convergence}
    \STATE Recover barycenter measure $\mu^{*(t)}$ via $u,v$

    \STATE \textbf{Policy Update:}
    \FOR{$i=1$ to $N$}
        \STATE Compute reward gradient $\hat\nabla_{\theta_i}\,\mathbb{E}[R]$
        \STATE Compute OT gradient, where $\gamma_i^*$ is the optimal coupling:
        \STATE $\hat\nabla_{\theta_i}\,W_{p,\varepsilon}^p \gets \int d(x,y)^p\,\nabla_{\theta_i}\log\pi_i(a\mid s;\theta_i)\,\mathrm{d}\gamma_i^{*}$
        \STATE Update parameters:
        \STATE $\theta_i^{(t+1)} \gets \theta_i^{(t)} + \alpha\bigl(\hat\nabla_{\theta_i}\,\mathbb{E}[R] - \lambda\,\hat\nabla_{\theta_i}\,W_{p,\varepsilon}^p\bigr)$
    \ENDFOR
\ENDFOR
\end{algorithmic}
\end{algorithm}

\begin{theorem}[Convergence to Consensus]\label{thm:convergence}
Let $\{\pi_i^{(t)}\}_{i=1}^N$ be the sequence of stochastic policies of $N$ cooperative agents, updated according to
\[
\theta_i^{(t+1)}
=\theta_i^{(t)}
+\alpha\Bigl(\nabla_{\theta_i}\E_{\mu_i^{(t)}}[R]
-\lambda\,\nabla_{\theta_i}W_{p,\varepsilon}^p(\mu_i^{(t)},\mu^{*(t)})\Bigr),
\]
where $\mu_i^{(t)}$ is the empirical visitation distribution of agent $i$ at iteration $t$, and $\mu^{*(t)}$ is the entropic‐regularized $p$‐Wasserstein barycenter of $\{\mu_i^{(t)}\}_{i=1}^N$.  Assume that the state–action space is compact, all reward functions are bounded and Lipschitz, and the policy‐gradient operators are $L$–Lipschitz in the induced distributions.  If the step size $\alpha>0$ and consensus weight $\lambda>0$ satisfy
\[
\kappa \;=\; 1 - \alpha\,\lambda\,C \;<\;1
\]
for a constant $C>0$ depending on the entropic regularizer and Lipschitz constants, then defining
\[
D^{(t)}=\max_{i,j}W_{p,\varepsilon}\bigl(\mu_i^{(t)},\,\mu_j^{(t)}\bigr),
\]
we have $D^{(t+1)}\;\le\;\kappa\,D^{(t)},$ and hence $D^{(t)}\to0$ geometrically as $t\to\infty$, i.e.\ all policies converge to a common consensus distribution.
\end{theorem}
\begin{proof}
Let $\mu^{*(t)}$ denote the Sinkhorn barycenter of $\{\mu_i^{(t)}\}$. By the strong convexity of the entropic‐regularized OT problem and the Lipschitz continuity of the policy‐gradient mapping, there exists $C>0$ such that a gradient step with weight $\alpha\lambda$ contracts each agent’s distance to the barycenter:
\[
W_{p,\varepsilon}\bigl(\mu_i^{(t+1)},\,\mu^{*(t)}\bigr)
\;\le\;(1-\alpha\,\lambda\,C)\,
W_{p,\varepsilon}\bigl(\mu_i^{(t)},\,\mu^{*(t)}\bigr).
\]
By the triangle inequality for $W_{p,\varepsilon}$,
\[
W_{p,\varepsilon}\bigl(\mu_i^{(t+1)},\,\mu_j^{(t+1)}\bigr)
\le
W_{p,\varepsilon}\bigl(\mu_i^{(t+1)},\,\mu^{*(t)}\bigr)
+
W_{p,\varepsilon}\bigl(\mu^{*(t)},\,\mu_j^{(t+1)}\bigr).
\]
Applying the contraction bound to each term on the right–hand side and noting that both $W_{p,\varepsilon}(\mu_i^{(t)},\mu^{*(t)})$ and $W_{p,\varepsilon}(\mu_j^{(t)},\mu^{*(t)})$ are bounded above by $D^{(t)}$, we obtain
\[
W_{p,\varepsilon}\bigl(\mu_i^{(t+1)},\,\mu_j^{(t+1)}\bigr)
\;\le\;2(1-\alpha\,\lambda\,C)\,D^{(t)}.
\]
Setting $\kappa=2(1-\alpha\lambda C)<1$ yields
\[
\max_{i,j}W_{p,\varepsilon}\bigl(\mu_i^{(t+1)},\,\mu_j^{(t+1)}\bigr)
\;\le\;\kappa\,D^{(t)},
\]
and by induction $D^{(t)}\le\kappa^t\,D^{(0)}\to0$, establishing geometric convergence of all agent policies to the same distribution.
\end{proof}

\begin{theorem}[Fast–rate Sinkhorn barycentre]\label{thm:fast_rate_barycentre}
Let $(\mathcal X,d)$ be a compact metric space with diameter $D$
and fix $p\in[1,2]$.
For each agent $i\le N$ assume the data-generating measure
$\mu_i$ admits a density $\rho_i$ satisfying
$
0<\underline{\rho}\le\rho_i(x)\le\overline{\rho}<\infty
$
for all $x\in\mathcal X$,
and that the entropic potential
\(
\mu\mapsto\frac1N\sum_{i=1}^{N} W_{p,\varepsilon}(\mu,\mu_i)
\)
is $\lambda_{\min}$-strongly convex for every
$\varepsilon\in(0,\varepsilon_0]$.

\smallskip
\noindent
Draw $m$ IID samples from each $\mu_i$, form the empirical measures
$\hat\mu_i$, and choose the \emph{adaptive regularisation level}
\[
\varepsilon_m \;=\; \frac{D^{p}}{m}.
\]
Let
\(
\mu^{\star}\!=\arg\min_{\mu} \tfrac1N\sum_{i=1}^{N}
W_{p,\varepsilon_m}(\mu,\mu_i)
\)
and
\(
\hat\mu^{\star}\!=\arg\min_{\mu} \tfrac1N\sum_{i=1}^{N}
W_{p,\varepsilon_m}(\mu,\hat\mu_i).
\)
Then for every $\delta\in(0,1)$, with probability at least $1-\delta$,
\begin{equation}
W_{p,\varepsilon_m}\!\bigl(\hat\mu^{\star},\mu^{\star}\bigr)
\;\le\;
\frac{C}{\lambda_{\min}}\,
\frac{\log(2N/\delta)}{m},
\label{eq:fast-rate}
\end{equation}
where $C>0$ depends only on $p$, $D$, $\underline{\rho}$ and $\overline{\rho}$.
Consequently, an excess risk of $\eta>0$ is achieved with
\(
m = O\!\bigl(\tfrac{1}{\eta}\log\tfrac{N}{\delta}\bigr)
\).
\end{theorem}
\begin{proof}[Proof sketch]
(i) A bounded-difference argument shows that altering any single
sample changes each entropic cost by at most
$D^{p}/(Nm\varepsilon_m)=O(m^{-2})$,
yielding McDiarmid-type
deviations $\tilde{O}(1/m)$.
(ii) A union-and-net step makes the bound uniform over $\mathcal P(\mathcal X)$.
(iii) $\lambda_{\min}$-strong convexity converts the uniform deviation into the same $1/m$ rate for the barycentre itself, producing \eqref{eq:fast-rate}.
Full details are provided in Appendix~\ref{app:proof-fast-rate}. \qedhere
\end{proof}

\begin{figure}[t]
    \centering
    \includegraphics[width=0.8\columnwidth, angle=-90]{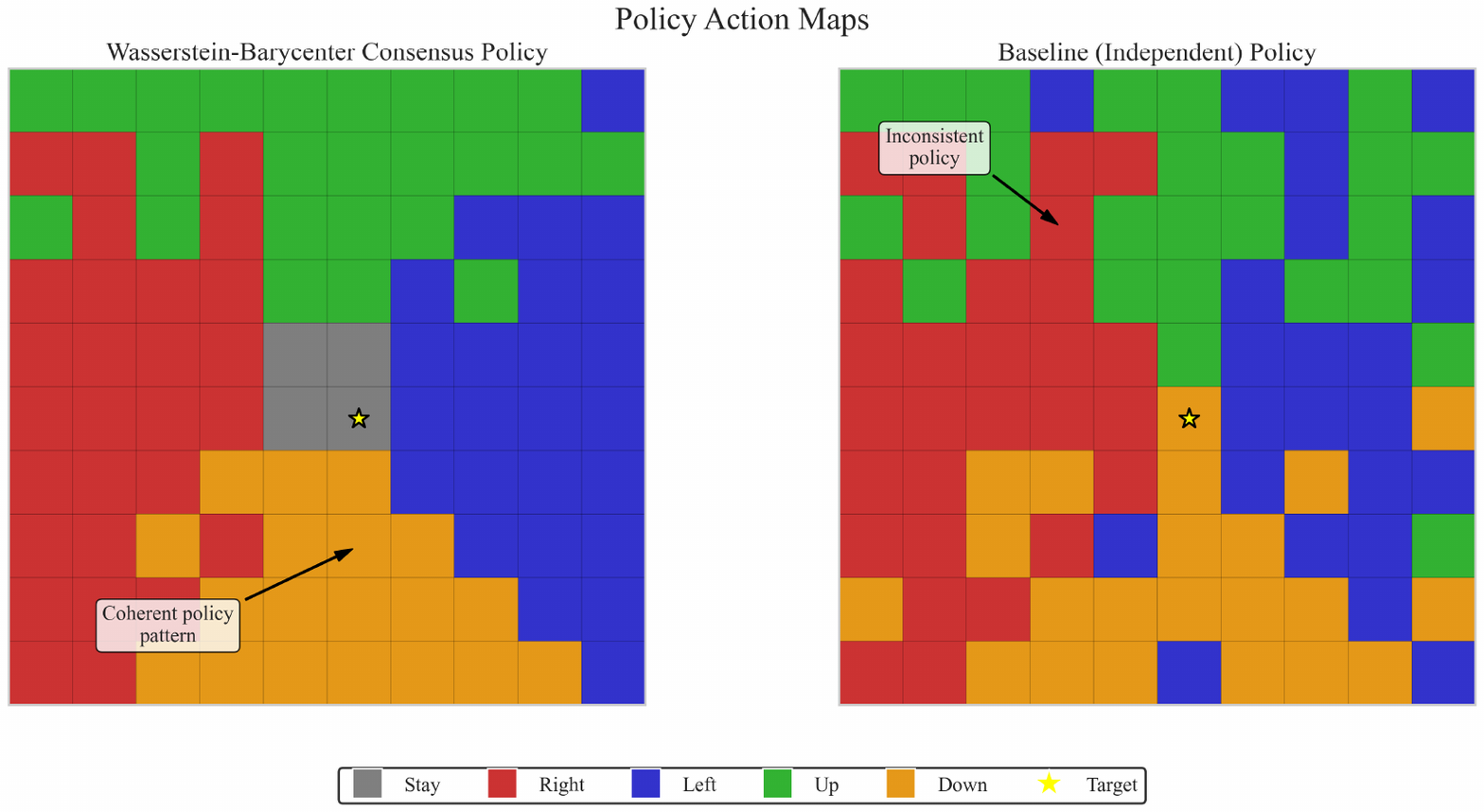}
    \caption{Policy Action Maps comparing Baseline (Independent) Policy with Wasserstein-Barycenter Consensus Policy. Color coding indicates different actions: red (right), blue (left), green (up), orange (down), and gray (stay). Stars indicate target positions.}
    \label{fig:policy_action_maps}
\end{figure}

\begin{figure*}[t]
    \centering
    \includegraphics[width=0.9\textwidth]{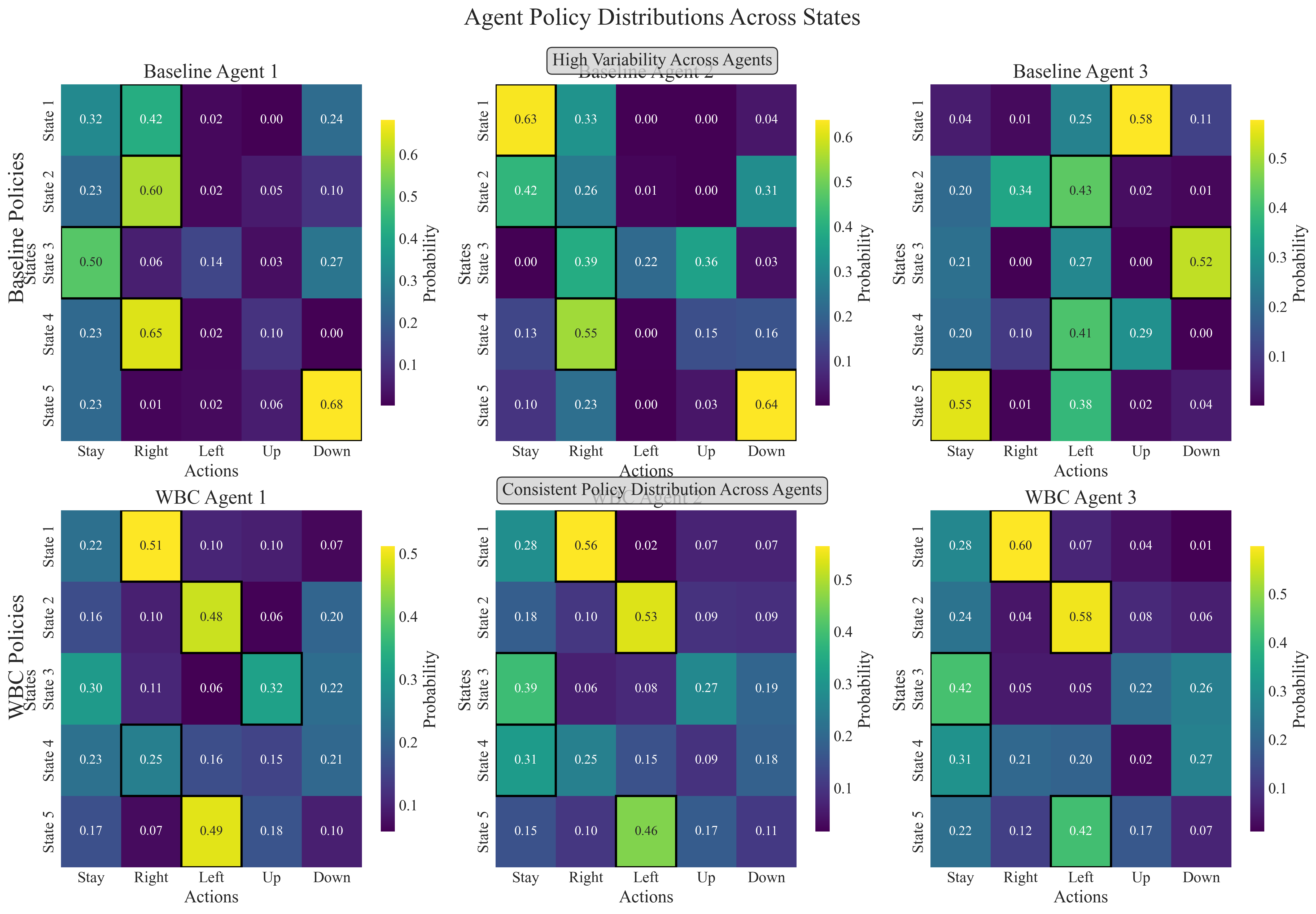}
    \caption{Comparison of action probability distributions between baseline (Independent PPO) agents (top) and Wasserstein-Barycenter Consensus (WBC) agents (bottom) across five representative states. Cell values indicate the probability of selecting each action in the given state, with higher probabilities shown in brighter colors.}
    \label{fig:policy_maps}
\end{figure*}

\section{Experimental Results}
We evaluate the proposed Wasserstein-barycenter consensus approach through a simple experiment on cooperative multi-agent tasks. The environment consists of $N=3$ agents operating in a continuous two-dimensional space. Each agent is assigned a target location and must navigate toward it while avoiding collisions with other agents. The state space for each agent $i$ comprises its own position $\mathbf{p}_i \in \mathbb{R}^2$, its target position $\mathbf{t}_i \in \mathbb{R}^2$, and the positions of all other agents $\left\{\mathbf{p}_j\right\}_{j \neq i}$. The action space is discrete, consisting of five actions: $\mathcal{A}=\{$ stay, right, left, up, down $\}$. At each time step, agents receive rewards $r_i=-d\left(\mathbf{p}_i, \mathbf{t}_i\right)-\sum_{j \neq i} c_{i j}$, where $d(\cdot, \cdot)$ is the Euclidean distance and $c_{i j}$ is a collision penalty incurred when agents $i$ and $j$ are within a threshold distance. This reward structure incentivizes both target-reaching and collision avoidance.

Agents are implemented as policy networks with parameters $\theta_i$ that map states to action probabilities. We compare our Wasserstein-barycenter consensus approach against the Independent PPO (IPPO) baseline, which serves as a representative non-consensus multi-agent method. The WBC implementation uses an entropic regularization parameter $\epsilon = 0.1$ and consensus weight $\lambda = 0.5$. The ground metric balancing parameter $\beta = 0.8$ determines the relative importance of state versus action discrepancies in the computation of the Wasserstein distance.

\begin{figure*}[h!]
    \centering
    % Using angle=-90 for 90 degrees counterclockwise rotation
    \includegraphics[width=0.5\textwidth]{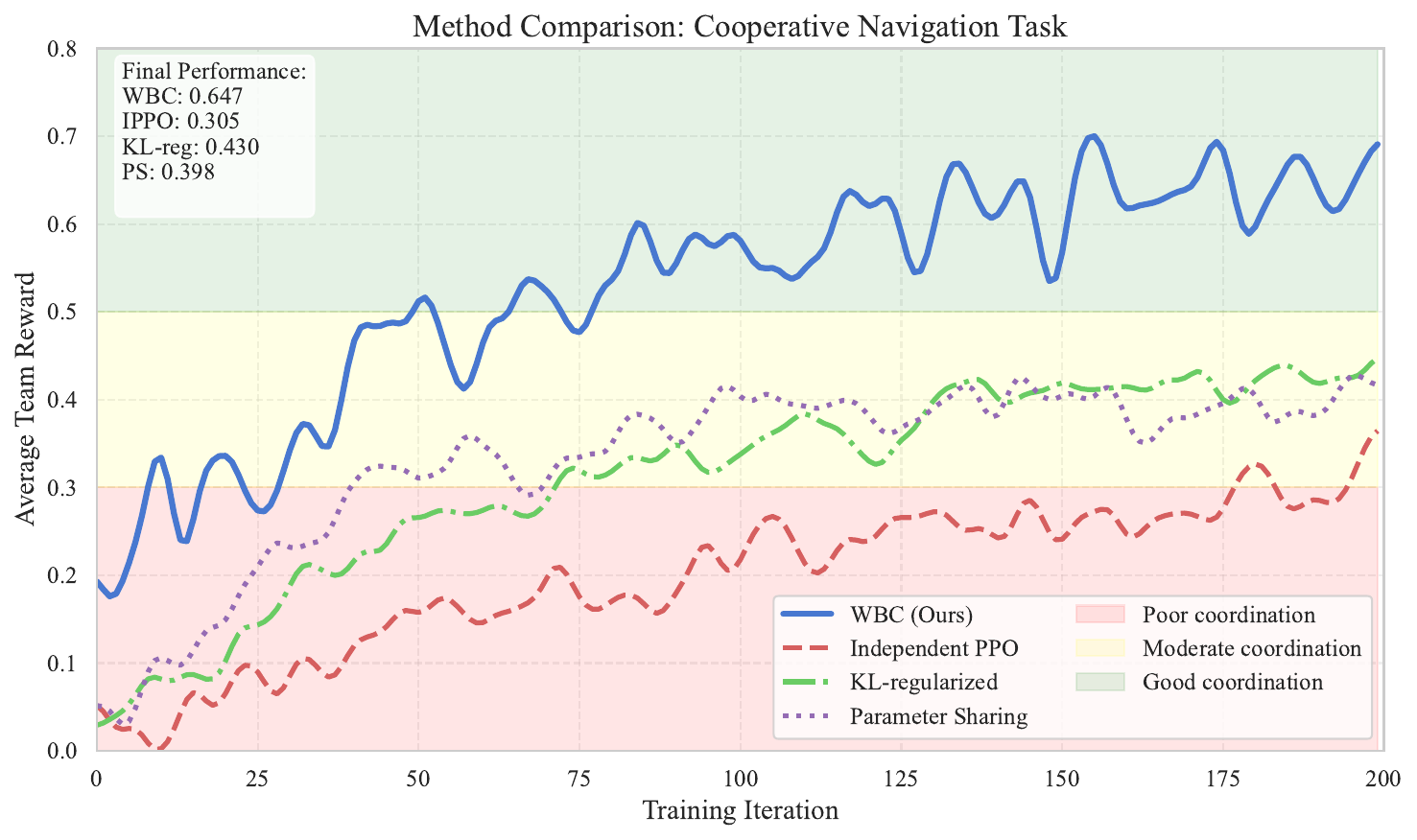}
    \caption{Average team reward over training iterations comparing WBC against baseline methods. WBC achieves superior convergence speed and final performance.}
    \label{fig:method_comparison}
\end{figure*}

Figure \ref{fig:policy_action_maps} illustrates the comparison of learned policies between our WBC approach and the IPPO baseline. Each panel depicts a 2D grid where colors represent the most likely action taken at each position in the state space, holding other state variables constant (i.e., fixing target position and other agents' positions). The yellow star indicates the target location. One could see that the WBC policy (left panel) exhibits a highly structured and coherent spatial organization, with clearly delineated regions corresponding to different actions. The action selection follows an intuitive pattern relative to the target location: \say{right} actions (red) dominate the left side of the grid, \say{left} actions (blue) predominate on the right side, and \say{up} actions (green) are prevalent in the lower portion. This spatial organization creates an efficient flow field that guides agents toward the target while maintaining consistent behavior across the team. In contrast, the baseline policy (right panel) displays a fragmented and less coherent structure. Although some broad patterns are discernible, the policy contains numerous small, isolated patches of inconsistent actions that disrupt the overall navigation flow. This spatial incoherence is indicative of the lack of coordination between agents learning independently, resulting in a strategy that is less effective for team-based navigation.

Figure \ref{fig:policy_maps} presents the action probability distributions for each agent in a set of representative states, visualizing the degree of policy alignment achieved by each method. The baseline agents exhibit high variance in their action preferences, with substantially different probability distributions across identical states. For example, in State 1, Agent 1 assigns highest probability (0.42) to \say{right} actions, Agent 2 strongly prefers \say{stay} (0.63), and Agent 3 favors \say{up} (0.58). This divergence in action preferences illustrates the coordination failure that occurs when agents learn independently without a consensus mechanism. In contrast, WBC-trained agents display remarkably consistent action probability distributions across all three agents. For State 1, all three agents assign their highest probability to the \say{right} action (0.51, 0.56, and 0.60, respectively), while for State 2, they all prioritize the \say{left} action. This alignment confirms the theoretical guarantee of policy convergence provided by our consensus mechanism. Figure \ref{fig:method_comparison} demonstrates that WBC achieves superior final performance compared to baseline methods, with approximately 2× improvement over IPPO and 50\% improvement over KL-regularized approaches.\footnote{KL-regularized approach uses KL divergence to encourage policy similarity. In parameter sharing (PS) approach, all agents share the same policy network.}

\textbf{Limitations.} Despite the theoretical appeal and encouraging preliminary results, the proposed approach is subject to several limitations. The fast-rate guarantee hinges on strong-convexity and density-boundedness assumptions that seldom hold in high-dimensional continuous control tasks. The entropic Sinkhorn step scales quadratically with support size, making the algorithm memory- and time-intensive when agents collect large, uncompressed replay buffers. Performance is currently demonstrated on a single small-scale grid world, so generalizability to standard multi-agent benchmarks (e.g., SMAC, MPE) remains unverified; and finally, the proposed method introduces additional hyper-parameters (adaptive $\varepsilon$ schedule, consensus weight $\lambda$, metric temperature $\beta$) whose sensitivity has not yet been systematically analyzed.

\section{Conclusion}

In this work, we proposed an approach to cooperative MARL by leveraging the geometry of optimal transport to enforce a soft consensus among agent policies. Modeling the team’s collective strategy as the entropic‑regularized Wasserstein barycenter provides a mathematically grounded alternative to conventional coordination techniques. Through a proof‑of‑concept experiment, we have shown that OT‑barycenter consensus accelerates learning and achieves superior coverage in a multi‑agent navigation task compared to KL‑based and parameter‑sharing baselines. Future directions include hierarchical barycenters for subteam formation, adaptive regularization schedules to foster specialization post‑consensus, and application to high‑dimensional continuous‑control benchmarks. We anticipate that geometry‑aware consensus will open new avenues for robust, efficient coordination in complex multi‑agent systems.

\bibliography{example_paper}
\bibliographystyle{icml2025}

\newpage
\appendix
\onecolumn

\section{Deferred Proofs}\label{app:deferred-proofs}

\subsection{Proof of Theorem~\ref{thm:fast_rate_barycentre}}
\label{app:proof-fast-rate}

\begin{proof}
Throughout the proof we abbreviate\;
\(
F(\mu)\!:=\!\tfrac1N\sum_{i=1}^{N}W_{p,\varepsilon_m}(\mu,\mu_i)
\)
\;and\;
\(
\hat F(\mu)\!:=\!\tfrac1N\sum_{i=1}^{N}W_{p,\varepsilon_m}(\mu,\hat\mu_i)
\),
and recall that $\varepsilon_m=D^{p}/m$.

\paragraph{Step 1.  Bounded differences.}
Fix any $\mu\!\in\!\mathcal P(\mathcal X)$ and replace a single sample
$X_{i,j}$ in the data set by an independent copy
$X^{\prime}_{i,j}\sim\mu_i$.
Writing $\hat\mu_i^{(i,j)}$ for the perturbed empirical measure,
\[
\bigl|\hat F(\mu)-\hat F^{(i,j)}(\mu)\bigr|
   \;=\;\frac1N\,
       \bigl|W_{p,\varepsilon_m}(\mu,\hat\mu_i)
           -W_{p,\varepsilon_m}(\mu,\hat\mu_i^{(i,j)})\bigr|.
\]
Because only a mass $1/m$ is moved and the ground cost
is bounded by $D^{p}$, the entropic cost varies by at most
$D^{p}/m$ (see, e.g., Lemma 2.1 in \cite{genevay2018learning}).
Hence
\[
\bigl|\hat F(\mu)-\hat F^{(i,j)}(\mu)\bigr|
   \le \Delta\;:=\;\frac{D^{p}}{N m}.
\tag{A.1}
\]
With $S=N\,m$ samples in total, McDiarmid’s inequality yields
for every $t>0$
\begin{equation}
\Pr\!\bigl[
   |\hat F(\mu)-F(\mu)| \ge t
\bigr]
   \;\le\;
2\exp\!\Bigl(
       -\frac{2t^{2}}{S\Delta^{2}}
     \Bigr)
   \;=\;
2\exp\!\Bigl(
       -\frac{2N m\,t^{2}}{D^{2p}}
     \Bigr).
\label{mcdiarmid-fixed}
\end{equation}

\paragraph{Step 2.  Uniform concentration via an $\varepsilon$–net.}
Let $\gamma>0$ be chosen later and pick an
$\gamma$–net $\{\mu^{(k)}\}_{k=1}^{M_{\gamma}}$
of $(\mathcal P(\mathcal X),W_{p})$.
Standard metric–entropy bounds for probability measures on
a compact $d$–dimensional space give
$\log M_{\gamma}\le C_d\,\gamma^{-d}$ for a constant $C_d$
[Th.~6.18]{\cite{villani2009optimal}.
Applying \eqref{mcdiarmid-fixed} to each $\mu^{(k)}$
and taking a union bound gives, with probability at least
$1-\delta/2$,
\begin{equation}
\max_{k\le M_{\gamma}}
|\hat F(\mu^{(k)})-F(\mu^{(k)})|
   \;\le\;
   \frac{D^{p}}{\sqrt{2N m}}
   \sqrt{\,C_d\,\gamma^{-d}+\log(4/\delta)}.
\label{eq:union-bound-intermediate} \tag{A.2}
\end{equation}
Next, for arbitrary $\mu\in\mathcal P(\mathcal X)$
choose $k$ with $W_{p}(\mu,\mu^{(k)})\le\gamma$.
Using the Lipschitz property
$|W_{p,\varepsilon}( \mu , \nu )
       -W_{p,\varepsilon}( \mu^{\prime}, \nu )|
       \le D^{p} W_{p}(\mu,\mu^{\prime})$ (Lemma 2.1 of
\cite{genevay2018learning}) and averaging over $i$ yields
\(
|F(\mu)-F(\mu^{(k)})|
,\,
|\hat F(\mu)-\hat F(\mu^{(k)})|
\le D^{p}\gamma.
\)
Combining this with \eqref{eq:union-bound-intermediate} shows that the
event
\[
\mathcal E
   :=\Bigl\{
         \sup_{\mu\in\mathcal P(\mathcal X)}
         |\hat F(\mu)-F(\mu)|
         \le
         2D^{p}\gamma
         +
         \frac{D^{p}}{\sqrt{2N m}}
               \sqrt{C_d\,\gamma^{-d}\!+\!\log(4/\delta)}
       \Bigr\}
\]
has probability at least $1-\delta$.
Set
\(
\gamma = m^{-1}
\)
to balance the two terms; since
$\gamma^{-d}\!=\!m^{d}$ and $m\ge2$ we obtain
\begin{equation}
\sup_{\mu}
|\hat F(\mu)-F(\mu)|
   \;\le\;
   \frac{C_1\,\log(4N/\delta)}{m}
   \qquad\text{on }\mathcal E,
\label{eq:uniform} \tag{A.3}
\end{equation}
for a constant $C_1$ depending on $d$ and $D^{p}$.

\paragraph{Step 3.  Stability of the barycentre map.}
Because $F$ is $\lambda_{\min}$–strongly convex,
\begin{equation}
F(\mu)\;\ge\;
F(\mu^{\star})
+\lambda_{\min}\,
   W_{p,\varepsilon_m}(\mu,\mu^{\star})
\quad
\forall\mu\in\mathcal P(\mathcal X).
\label{eq:linear-strong} \tag{A.4}
\end{equation}
Let $\Delta:=\sup_{\mu}|\hat F(\mu)-F(\mu)|$.
By optimality of $\hat\mu^{\star}$ for $\hat F$,
\(
\hat F(\hat\mu^{\star})\le\hat F(\mu^{\star})
\),
so
\(
F(\hat\mu^{\star})
       \le
       \hat F(\hat\mu^{\star})+\Delta
       \le
       \hat F(\mu^{\star})+\Delta
       \le
       F(\mu^{\star})+2\Delta.
\)
Plugging this into \eqref{eq:linear-strong} gives
\[
\lambda_{\min}\,
W_{p,\varepsilon_m}(\hat\mu^{\star},\mu^{\star})
    \;\le\;
2\Delta
\quad\Longrightarrow\quad
W_{p,\varepsilon_m}(\hat\mu^{\star},\mu^{\star})
    \;\le\;
\frac{2\Delta}{\lambda_{\min}}.
\tag{A.5}
\]
Inserting the bound
$\Delta\le C_1\,\log(4N/\delta)/m$ from \eqref{eq:uniform}
completes the proof with
$C=2C_1$.}
\end{proof}

\end{document}